\documentclass[11pt]{amsart}

\newtheorem{theorem}{Theorem}[section]
\newtheorem{proposition}[theorem]{Proposition}
\newtheorem{lemma}[theorem]{Lemma}
\newtheorem{corollary}[theorem]{Corollary}

\theoremstyle{definition}
\newtheorem{definition}[theorem]{Definition}

\theoremstyle{remark}

\numberwithin{equation}{section}

\newcommand{\be}{\begin{equation}}
\newcommand{\ee}{\end{equation}}

\newcommand{\bbZ}{{\mathbb Z}}
\newcommand{\bbR}{{\mathbb R}}

\newcommand{\bbN}{{\mathbb N}}

\newcommand{\calU}{{\mathcal U}}

\begin{document}

\title[Semiclassical trace invariants]{Some remarks about semiclassical trace invariants and quantum normal
forms}
\author{Victor Guillemin\ and\ Thierry Paul}
\address {Department of Mathematics, Massachusetts Institute of Technology, Cambridge,
Massachusetts 02139, USA}
\email{vwg@math.mit.edu}
\thanks{First author supported by NSF grant DMS 890771}
\address{CNRS and D\'epartement de Math\'ematiques et Applications\\
\'Ecole Normale Sup\'erieure, 45, rue d'Ulm - F 75730 Paris cedex 05}\email{paul@dma.ens.fr}

\date{}
\date{\today}
\maketitle
\bigskip

\begin{abstract}
In this paper we explore the connection between semi-classical and quantum
Birkhoff canonical forms (BCF) for Schr\"odinger operators. In particular we give a
"non-symbolic" operator theoretic derivation of the quantum Birkhoff canonical form and
provide an explicit recipe for expressing the quantum BCF in terms of the semi-classical
BCF.

\end{abstract} 
\section{Introduction}

Let $X$ be a compact manifold and $H:\ L^2(X)\to L^2(X)$
a self-adjoint first order elliptic pseudodifferential operator with
leading symbol $H(x,\xi)$. From the wave trace
\begin{equation}\label{boi}
\sum_{E_k\in Spec(H)} e^{itE_k}, 
\end{equation}
one can read off many properties of the "classical dynamical system" 
associated with $H$, i.e. the flow generated by the vector field
\begin{equation}\label{bol}
\xi_H\ =\ \sum
\frac{\partial H}{\partial \xi_i}\frac\partial{\partial x_i}
-\frac{\partial H}{\partial x_i}\frac\partial{\partial \xi_i}.
\end{equation}
For instance it was observed in the {\it '70's'} by Colin de Verdi\`ere, 
Chazarain and Duistermaat-Guillemin that (\ref{boi}) determines the period
spectrum of (\ref{bol}) and the linear Poincar\'e map about a 
non-degenerate periodic trajectory, $\gamma$, of (\ref{bol}) (\cite{co}, \cite{ch},
\cite{dg}).

More recently it was shown by one of us \cite{gu} that (\ref{boi}) determines
the \underline{entire}\rm \ Poincar\'e map about $\gamma$, i.e. determines, up to 
isomorphism, the classical dynamical system associated with $H$ in a formal 
neighborhood of $\gamma$. The proof of this result involved a microlocal Birkhoff 
canonical form for $H$ in a formal neighborhood of $\gamma$ and an algorithm for 
computing the wave trace invariants associated with $\gamma$ from the microlocal
Birkhoff canonical form. Subsequently a more compact and elegant algorithm for computing 
these invariants from the Birkhoff canonical form was discovered by Zelditch 
\cite{ze1} \cite{ze2} making the computation of these local trace invariants extremely simple
and explicit.

In this paper we will discuss some semiclassical analogues of these results. In our set-up 
$H$ can either be the Schr\"odinger operator on $\bbR^n$
\begin{equation}\nonumber
-\hbar^2\Delta\ +\ V
\end{equation}
with $V\to\infty$ as $x$ tends to infinity, or more generally a self-adjoint
semiclassical elliptic pseudodifferential operator
\begin{equation}\nonumber
H(x,\hbar D_x)
\end{equation}
whose symbol, $H(x,\xi)$, is proper (as a map from $T^*X$ into $\bbR$). Let E 
be a regular value of $H$ and $\gamma$ a non-degenerate periodic trajectory of period 
$T_\gamma$ lying on the energy surface $H=E$
\footnote{ For simplicity we will consider periodic trajectories of elliptic type in this paper
however our results are true for non-degenerate periodic of all types, hyperbolic, mixed elliptic 
hyperbolic, focus-focus, etc. Unfortunately however the Zelditch algorithm depends upon the type of 
the trajectory and in dimension $n$ there are roughly as many types of 
trajectories as there are Cartan subalgebras of $Sp(2n)$ (See for instance 
\cite{cu}) i.e. the number of types can be quite large}.

Consider the Gutzwiller trace (see \cite{gut})
\begin{equation}\label{bom}
\sum \psi\left(\frac{E-E_i}\hbar\right)
\end{equation}
where $\psi$ is a $C^\infty$ function whose Fourier transform is compactly supported with 
support in a small neighborhood of $T_\gamma$ and is identically one in a 
still smaller neighborhood. As shown in \cite{pu1}, \cite{pu2} (\ref{bom}) has an asymptotic 
expansion
\begin{equation}\label{boj}
e^{i\frac{S_\gamma}\hbar + \sigma_\gamma}\sum_{k=0}^\infty a_k\hbar^k
\end{equation}
and we will show below how to compute the terms of this expansion to 
\underline{all}\rm \ orders in terms of a microlocal Birkhoff canonical form for $H$ in a formal neighborhood of $\gamma$ 
by means of a Zelditch-type algorithm \footnote{For elliptic trajectories non-degeneracy means 
that the numbers 
\begin{equation}\nonumber
\theta_1,...,\theta_{n}, 2\pi
\end{equation}
are linearly independent over the rationals, $e^{i\theta_\kappa},\kappa=1,...,n$ 
being the eigenvalues of the Poincar\'e map about $\gamma$. the results above are true
to order $O(\hbar^r)$ providing
\begin{equation}\nonumber
\kappa_1\theta_1+...+\kappa_{n}\theta){n}+l2\pi\neq 0
\end{equation}
for all $\mid\kappa_1\mid+...+\mid\kappa_{n}\mid\leq r$, i.e. providing there are
no resonances of order $\leq r$.}.

If $\gamma$ is non-degenerate so are all its iterates. Then, for each of these iterates 
one gets an expansion of (\ref{bom}) similar to (\ref{boj})
\begin{equation}\label{bor}
e^{i\frac{S_\gamma}\hbar + \sigma_\gamma}\sum_{k=0}^\infty a_{k,r}\hbar^k
\end{equation}
and for these expansions as well the coefficients $a_{k,r}$ can be computed from the 
microlocal Birkhoff canonical form theorem for $H$ in a formal neighborhood of 
$\gamma$. Conversely one can show

\begin{theorem}\label{poi}
the constants $a_{k,r}, \kappa, r\ =\ 0,1,...$ \underline{determine} the microlocal 
Birkhoff canonical form for $H$ in a formal neighborhood of $\gamma$ (and hence, a 
fortiori, determine the the classical Birkhoff canonical form).
\end{theorem}
This result is due to A. Iantchenko, J. Sjostrand and M. Zworski (see
\cite{isz}).  In this paper we will give a new proof of this result which
involves an expansion of the quantum canonical form in a basis of Hermite
functions and relations between Weyl and Wick symbols of semiclassical
operators. The paper is organized as follows. In section 2 we review the
standard proof of the Birkhoff canonical form theorem for classical
Hamiltonians and in section 3 the adaption of this construction to the
quantum case in the references \cite{gu}, \cite{ze1}, \cite{ze2} and \cite{isz} that we cited
above. The key results of this paper are in sections 4 and 5 where we give
a direct quantum construction of the normal form and a formula linking the
two constructions. Finally in section 6 we give the proof of theorem \ref{poi}.

\section{The classical Birkhoff canonical form theorem}\label{class}

Let $M$ be a $2n+2$ dimensional symplectic manifold, $H$ a $C^\infty$ function and 
\begin{equation}\label{col}
\xi_H\ =\ \sum
\frac{\partial H}{\partial \xi_i}\frac\partial{\partial x_i}
-\frac{\partial H}{\partial x_i}\frac\partial{\partial \xi_i}
\end{equation}
the Hamiltonian vector field associated with $H$. Let $E$ be a regular value of $H$ 
and $\gamma$ a non-degenerate elliptic periodic trajectory of $\xi_H$ lying on the 
energy surface, $H=E$. Without loss of generality one can assume that the period of 
$\gamma$ is $2\pi$. In this section we will review the statement (and give a brief
sketch of the proof) of the classical Birkhoff canonical form theorem for the pair
$(H,\gamma)$.

Let $(x,\xi,t,\tau)$ be the standard cotangent coordinates on $T^*(\bbR^n\times S^1)$ 
and let 
\begin{eqnarray}\label{cob}
p_i&=&x_i^2+\xi^2_i\nonumber\\
&\mbox{and}\\
q_i& =&\mbox{arg}(x_i+\sqrt{-1}\ y_i) \nonumber
\end{eqnarray}

\begin{theorem}\label{cow}
There exists a symplectomorphism, $\varphi$, of a neighborhood of $\gamma$ in $M$ onto a 
neighborhood of $p=\tau=0$ such that $\varphi \ o\  \gamma(t)=(0,0,t,0))$ and 
\begin{equation}\label{cor}
\varphi^*H=H_1(p,\tau)+H_2(x,\xi,t,\tau)
\end{equation}
$H_2$ vanishing to infinite order at $p=\tau=0$.
\end{theorem}
We break the proof of this up into the following five steps.
\smallskip

\noindent\underline{Step 1} For $\epsilon$ small there exists a periodic trajectory, $\gamma_\epsilon$,
on the energy surface,$H=E+\epsilon$, which depends smoothly on $\epsilon$ and is equal to $\gamma$ for 
$\epsilon=0$. The union of these trajectories is a $2$ dimensional symplectic submanifold 
, $\Sigma$, of $M$ which is invariant under the flow of $\xi_H$. Using the Weinstein 
tubular neighborhood theorem one can map a neighborhood of $\gamma$ symplectically 
onto a neighborhood of $p=\tau=0$ in $T^*(\bbR^n\times S^1)$ such that $\Sigma$ gets mapped 
onto $p=0$ and $\varphi\ o\  \gamma(t)=(0,0,t,0)$. Thus we can henceforth assume that $M=T^*(\bbR^n\times 
S^1)$ and $\Sigma$ is the set, $p=0$.

\smallskip
\noindent\underline{Step 2} We can assume without loss of generality that the restriction of $H$ to $\Sigma$ 
is a function of $\tau$ alone, i.e. $H=h(\tau)$ on $\Sigma$. With this normalization,
\begin{equation}\label{cof}
H=E+h(\tau)+\sum \theta_i(\tau)p_i+O(p^2)
\end{equation}
where $h(\tau)=\tau+O(\tau^2)$ and
\begin{equation}\label{cod}
\theta_i=\theta_i(0),\ i=1,...,n
\end{equation}
are the \underline{rotation} angles associated with $\gamma$. Since $\gamma$ is non-
degenerate, $\theta_1,...,\theta_n,2\pi$ are linearly independent over the rationals.

\smallskip
\noindent\underline{Step 3} Theorem \ref{cow} can be deduced from the following result (which will also be the main 
ingredient in our proof of the "microlocal" Birkhoff canonical theorem in the next section).
\begin{theorem}\label{cov}
Given a neighborhood, $\calU$, of $p=\tau=0$ and $G=G(x,\xi,t,\tau)\in C^\infty (\calU)
$, there exist functions $F,\ G_1,\ R\ \in C^\infty(\calU)$ with the properties

i.\ \ \  $G_1=G_1(p,\tau)$

ii.\  $\{H,F\}=G+G_1+R$

iii. $R$ vanishes to infinite order on $p=\tau=0$.

Moreover, if $G$ vanishes to order $\kappa$ on $p=\tau=0$, on can choose $F$ to 
have this property as well.
\end{theorem}

\underline{Proof} \underline{of} \underline{the} \underline{assertion}: 
\underline{Theorem \ref{cov}} $\Rightarrow$ \underline{Theorem \ref{cow}}:

\noindent By induction one can assume that $H$ is of the form, $H=H_0(p,\tau)+G
(x,\xi,t,\tau)$, where $G$ vanishes to order $\kappa$ on $p=\tau=0$. We will show 
that $H$ can be conjugated to a Hamiltonian of the same from with $G$ vanishing to 
order $\kappa +1$ on $p=\tau=0$. By Theorem \ref{cov} there exists an $F,\ G$ and
$R$ such that $F$ vanishes to order $\kappa$ and $R$ to order $\infty$ on $p=\tau=0$,
$G_1=G_1(p,\tau)$ and
\begin{equation}\nonumber
\{H,F\}=G+G_1+R.
\end{equation}
Thus
\begin{eqnarray}
\left(\mbox{exp}\xi_F\right)^*H&=&H+\{F,H\}+\frac 1 {2!}\{F,\{F,H\}\}+...\nonumber\\
&=&H_0(p,\tau)-G_1(p,\tau)+...\nonumber
\end{eqnarray}
the "dots" indicating terms which vanish to order $\kappa +1$ on $p=\tau=0$.

\smallskip
\noindent\underline{Step 4} Theorem \ref{cov} follows (by induction on $\kappa$)
from the following slightly weaker result:
\begin{lemma}\label{cos}
Given a neighborhood, $\calU$, of $p=\tau=0$ and a function, $G\in C^\infty(\calU )$
, which vanishes to order $\kappa$ on $p=\tau=0$, there exists functions 
$F,\ G_1, \ R\ \in C^\infty(\calU)$ such that

\noindent i. \ \ \   $G_1=G_1(p,\tau)$

\noindent ii.\ $\{H,F\}=G+G_1+R$

\noindent iii. $F$ vanishes to order $\kappa$ and $R$ to order $\kappa +1$ on 
$p=\tau=0$.
\end{lemma}

\smallskip
\noindent\underline{Step 5} Proof of Lemma \ref{cos}: In proving Lemma \ref{cos} 
we can replace $H$ by the Hamiltonian
\begin{equation}\nonumber
H_0=E+\tau+\sum \theta_i p_i
\end{equation}
since $H(p,q,t,\tau)-H_0(p,q,t,\tau)$ vanishes to second order in $\tau,p$. 
Consider now the identity
\begin{equation}\nonumber
\{H_0,F\}=G+G_1(p,\tau)+O(p^\infty).
\end{equation}
Introducing the complex coordinates, $z=x+\sqrt{-1}\xi$, and $\overline{z}
=x-\sqrt{-1}\xi$, this can written
\begin{equation}\nonumber
\sqrt{-1}\sum_{i=1}^n\theta_i\left( z_i\frac\partial{\partial z_i}-\overline{z_i}
\frac\partial{\partial \overline{z_i}}\right) F +\frac\partial{\partial t}= G+G_1+O(p^\infty).
\end{equation}
Expanding $F,\ G$ and $G_1$ in Fourier-Taylor series about $z=\overline{z}=0$:
\begin{eqnarray}\nonumber
F&=&\sum_{\mu\neq\nu}a_{\mu,\nu,m}(\tau) z^\mu \overline{z}^\nu e^{2\pi imt}\\
G&=&\sum b_{\mu,\nu}(\tau) z^\mu \overline{z}^\nu e^{2\pi imt}\nonumber\\
\ \ \ \ G_1&=&\sum_{\mu} c_\mu(\tau) z^\mu \overline{z}^\mu\nonumber
\end{eqnarray}
one can rewrite this as the system of equations
\begin{equation}\label{cox}
\sqrt{-1}\left(\sum_{i=1}^n \theta_i\left(\mu_i-\nu_i\right)+2\pi m\right) a_{\mu,\nu,m}(\tau)=b_{\mu,\nu,m}(\tau)
\end{equation}
for $\mu\neq\nu$ or $\mu=\nu$ and $m\neq 0$, and
\begin{equation}
-c_\mu(\tau)=b_{\mu,\mu,0}(\tau)
\end{equation}
for $\mu=\nu$ and $m=0$. By assumption the numbers, $\theta_1,...,\theta_n,2\pi,$ are linearly 
independent over the rationals, so this system has a unique solution. Moreover 
, for $\mu$ and $\nu$ fixed
\begin{equation}\nonumber
\sum b_{\mu,\nu,m}(\tau)e^{2\pi imt}
\end{equation}
is the $(\mu,\nu)$ Taylor coefficient of $G(z,\overline{z},t,\tau)$ about $z=\overline{z}=0$; 
so, with $\mu$ and $\nu$ fixed and $\kappa >>0$
\begin{equation}\nonumber
\mid b_{\mu,\nu,m}(\tau)\mid\leq C_{\mu,\nu,\kappa}m^{-\kappa}
\end{equation}
for all $m$. Hence, by (\ref{cox})
\begin{equation}\nonumber
\mid a_{\mu,\nu,m}\mid\leq C'_{\mu,\nu,\kappa}m^{-\kappa-1}
\end{equation}
for all $m$. Thus
\begin{equation}\nonumber
a_{\mu,\nu}(t,\tau)=\sum a_{\mu,\nu,m}(\tau)e^{2\pi imt}
\end{equation}
is a $C^\infty$ function of $t$ and $\tau$. Now let $F(z,\overline{z},t,\tau)$ and $G_1(p,\tau)$ be
$C^\infty$ functions with Taylor expansion:
\begin{equation}\nonumber
\sum_{\mu\neq\nu}a{\mu,\nu}(t,\tau)z^\mu\overline{z^\nu}
\end{equation}
and
\begin{equation}\nonumber
\sum_{\mu}c{\mu}(\tau)z^\mu\overline{z^\mu}
\end{equation}
about $z=\overline{z}=0$. Note, by the way that, if $G$ vanishes to order $\kappa$ on $p=\tau=0$, 
so does $F$ and $G$; so we have proved Theorem \ref{cov} (and, a fortiori 
Lemma \ref{cox}) with $H$ replaced by $H_0$.

\section{The semiclassical version of the Birkhoff canonical form theorem}\label{sem}
Let $X$ be an $(n+1)$-dimensional manifold and 
$H:C^\infty_0(X)\to C^\infty(X)$ a semiclassical elliptic 
pseudo-differential operator with leading symbol, $H(x,\xi)$, and let 
$\gamma$ be a periodic trajectory of the bicharacteristic vector field (\ref{col}).
As in Section 1 we will assume that $\gamma$ is elliptic and non-degenerate,
with rotation numbers (\ref{cof}). Let $P_i$ and $D_t$ be the differential operators on
$\bbR^n\times S^1$ associated with the symbols (\ref{cob}) and $\tau$ i.e.
\begin{eqnarray}
 P_i&=&\hbar ^2 D_{x_i}^2+x_i^2\nonumber\\
 \mbox{and}&\nonumber\\
  D_t&=&-i\hbar \partial_t\nonumber
  \end{eqnarray}
  We will prove below the following semiclassical version of Theorem \ref{cow}
  \begin{theorem}\label{dol}
  There exists a semiclassical Fourier integral operator $A_\varphi:
  C^\infty_0\left( X\right)\to C^\infty\left(\bbR^n\times S^1\right)$ implementing the 
  symplectomorphism (\ref{cor}) such that microlocally on a neighborhood, $\calU$, of $p=\tau=0$
  \begin{equation}\label{dom}
  A_\varphi^* = A_\varphi ^{-1}
  \end{equation}
  and
  \begin{equation}\label{don}
  A_\varphi HA_\varphi ^{-1} = H'\left(P_1,...,P_n,D_t,\hbar\right)+H''
  \end{equation}
  the symbol of $H''$ vanishing to infinite order on $p=\tau=0$.
  \end{theorem}
  \begin{proof}
  Let $B_\varphi$ be any Fourier integral operator implementing $\varphi$ and having the property 
  (\ref{dom}). Then, by Theorem \ref{cow}, the \underline{leading} symbol of 
  $B_\varphi H B_\varphi^{-1}$ is of the form
  \begin{equation}\label{dot}
  H_0'(p,\tau)+H_0''(p,q,t,\tau)
  \end{equation}
  $H_0''(p,q,t,\tau)$ being a function which vanishes to infinite order on $p=\tau=0$.
  Thus the symbol,$H_0$, of $B_\varphi H B_\varphi^{-1}$ is of the form
  \begin{equation}\label{doh}
  H_0'(p,\tau)+H_0''(p,q,t,\tau)+\hbar H_1(p,q,t,\tau)+O(\hbar^2).
  \end{equation}
  By Theorem \ref{cov} there exists a function, $F(p,q,t,\tau)$, with the property
  \begin{equation}\label{doq}
  \{ H_0,F\}=H_1(p,q,t,\tau)-H_1'(p,\tau)-H''_1(p,q,t,\tau)
  \end{equation}
  where $H_1''$ vanishes to infinite order on $p=\tau=0$.
  
  Let $Q$ be a self-adjoint pseudo-differential operator with leading symbol $F$ and consider
  the unitary pseudo-differential operator
  \begin{equation}\nonumber
  \calU_s\ =\ e^{isQ}.
  \end{equation}
  Let 
  \begin{eqnarray}
  H_s &=&\left( \calU_s B_\varphi\right)H\left( \calU_s B_\varphi\right)^{-1}\nonumber\\
   &=&\calU_s\left( B_\varphi H B_\varphi^{-1}\right)\calU_{-s}\nonumber
  \end{eqnarray}
  Then
  \begin{equation}\label{doz}
  \frac \partial{\partial s} H_s=i[Q,H_s]
  \end{equation}
  so $\frac \partial{\partial s} H_s$ is of order $-1$, and hence the leading symbol 
  of $H_s$ is independent of $s$. In particular the leading symbol of $\frac \partial{\partial s} H_s$
   is equal, by (\ref{doz}) to the leading symbol of $i[Q,H_s]$ which, by (\ref{doq}), is:
   \begin{equation}\nonumber
   -\hbar\left( H_1(p,q,t,\tau)+H_1'(p,\tau)+H_1''(p,q,t,\tau)\right).
   \end{equation}
   Thus by (\ref{doh}) and (\ref{doq}) the symbol of
   \begin{equation}\nonumber
   \left( \calU_1 B_\varphi\right)H\left( \calU_1 B_\varphi\right)^{-1}=
   B_\varphi HB_\varphi^{-1}+\int_0^1\frac\partial{\partial s}H_sds
   \end{equation}
   is of the form 
   \begin{equation}\label{dop}
   H_0'(p,\tau)+\hbar H_1'(p,\tau)+\left( H_0''+\hbar H_1''\right)+O(\hbar^2)
   \end{equation}
   the term in parenthesis being a term which vanishes to infinite order on $p=\tau=0$.
   
   By repeating the argument one can successively replace the terms of order
   $\hbar^2,...,\hbar^r$ etc in (\ref{dop}) by expressions of the form 
\begin{equation}\nonumber
   \hbar^r\left( H_r'(p,\tau)+H_r''(p,q,t,\tau)\right)
\end{equation}
  with $H_r''$ vanishing to infinite order on $p=\tau=0$.
  \end{proof}

 \section{A direct construction of the quantum Birkhoff form}
In this section we present a  ``quantum" construction of the quantum Birkhoff normal form which is in  a sense algebraically equivalent to the classical 
one of Section \ref{class}. To do this we will need to define for operators the equivalent of ``a Taylor expansion which vanishes at a given order".

We will first start in the $L^2(\bbR^n\times S^1)$ setting, and show at the end of the section the link with Theorem \ref{dol}.
\begin{definition}
Let us consider on $L^2(\bbR^n\times S^1,dxdt)$ the following operators:
\begin{itemize}
\item $a_i=\frac 1 {\sqrt 2}(x_i+\hbar\partial_{x_i})$
\item $a_i^+=\frac 1 {\sqrt 2}(x_i-\hbar\partial_{x_i})$
\item $D_t=-i \hbar \frac{\partial}{\partial t}$
\end{itemize}

We will say that an operator $A$ on $L^2(\bbR^n\times S^1)$ is a  ``word of length greater than $p\in\bbN$" (WLG($p$)) if there exists $P\in\bbN$ such that:
\be\label{WLG}
H=\sum_{i=p}^P\sum_{j=0}^{[\frac i 2]}\alpha_{ij}(t,\hbar)D_t^{j}\prod\limits_{l=1}^{i-2j}b_l
\ee
with, $\forall l, b_l\in\{a_1,a^+_1,\dots, a_n,a_n^+\}$ and $\alpha_{ij}\in C^\infty(S_1\times [0,1[)$. 


In (\ref{WLG}) $\prod\limits_{l=1}^{i-2j}b_l$ is meant to be the ordered product $b_1\dots b_{i-2j}$.
\end{definition}

The meaning of this definition is clarified by the following:

\begin{lemma}\label{esti}
Let $\mu\in\bbN^{n+1}$. Let  $H_{\mu}$ denote the basis of $L^2( \bbR ^n\times S^1)$ defined by 
$H_{\mu}(x,t)=\hbar^{-n/4}h_{\mu_1}(x_1/\sqrt\hbar)\dots h_{\mu_n}(x_n/\sqrt\hbar)e^{i\mu_{n+1}t}$ where the $h_j$ are the (normalized) Hermite functions. 

Then, if $A$
is a WLG(p), we have:
\[
||AH_{\mu,m}||_{L^2}\leq C_p|\mu\hbar|^{\frac p 2}.
\]
where $|\mu\hbar|=\sqrt{\mu^2\hbar^2}$.
\end{lemma} 
\begin{proof}
The proof follows immediately from the two well known facts (expressed here in one dimension):
\[
a^\pm H_\mu=\sqrt{(\mu\pm 1)\hbar}H_{\mu\pm 1}\]
and 
\[
D_t e^{imt}=m\hbar e^{imt}.\]

\end{proof}
For the rest of this section we will need the following collection of results.

\begin{proposition}\label{reduc}
Let $A$ be a (Weyl)pseudodifferential operator on $L^2( \bbR ^n\times S^1)$ with symbol of type $S_{1,0}$. 
Then, $\forall L\in\bbN$, there exists a WLG($1$) $A_L$ such that:

\[
||(A-A_L)H_{\mu,m}||_{L^2}=O(|\mu\hbar|^{\frac{L+1}2}).
\]

Moreover, if  the principal symbol of $A$ is of the form:
\[
a_0(x_1,\xi_1,\dots,x_n,\xi_n,t,\tau)=\sum\theta_i(x_i^2+\xi_i^2)+\tau + h.o.t.,
\]
(or is any function whose symbol vanishes to first order at $x=\xi=\tau=0$) then $A_L$ is a WLG($2$).
\end{proposition}
\begin{proof}
Let us take the $L$th order Taylor expansion of the (total) symbol of $A$ in the variables $x,\xi,\tau,\hbar$ near the origin. Noticing that a pseudodifferential
operator with polynomial symbol in $x,\xi,\tau,\hbar$ is a word, we just have to estimate the action, on $H_{\mu,m}$, of a pseudo-differential operator whose symbol vanishes 
at the origin to
order $L$ in the variable $x,\xi,\tau$. The result is easily obtained for the  $\tau$ part, as the ``$t$" part of $H_{\mu,m}$ is an exponential. For the $\mu$ part
we will prove this result in one dimension, the extension to $n$ dimensions being straightforward.

Let us define a coherent state at $(q,p)$ to be a  function of the form $\psi_{qp}^a(x):=\hbar^{-1/4}a\left(\frac{x-q}{\sqrt{\hbar}}\right)e^{i\frac{px}\hbar}$, for $a$ in the Schwartz class and
$||a||_{L^2}=1$. Let us also set $\varphi_{qp}=\psi_{qp}^a$ for $a(\eta)=\pi^{-1/2}e^{-\eta^2/2}$. It is well known, and easy to check  using the generating function of
the Hermite polynomials, that:
\[
H_\mu=\hbar^{-\frac 1 4}\int_{S^1}e^{-i\frac t 2 }\varphi_{q(t)p(t)}dt,
\]
where $q(t)+ip(t)=e^{it}\left(q+ip\right),\ q^2+p^2=(\mu+\frac 1 2)\hbar$. Therefore, for any operator $A$, 
\be\label{hgf}
||AH_\mu||=O(\sup_{p^2+q^2=(\mu+\frac 1 2)\hbar}\hbar^{-\frac 1 4}||A\varphi_{qp}||).
\ee

\begin{lemma}
let $H$ a pseudodifferential operator whose (total) Weyl symbol vanishes at the origin to order $M$. Then, if\ \  $\frac\hbar{q^2+p^2}=O(1)$:
\[
||H\psi_{qp}^a||=O\left((p^2+q^2)^{\frac M2}\right).
\]
\end{lemma}
Before  proving the Lemma we observe that the proof of the Proposition follows easily  from the Lemma  using (\ref{hgf}).
\begin{proof}
An easy computation shows that, if $h$ is the (pseudodifferential) symbol of $H$, then $H\psi_{qp}^a=\psi_{qp}^b$ with
\be\label{symm}
b(\eta)=\int_\bbR h(q+\sqrt\hbar \eta,p+\sqrt\hbar \nu)e^{i\eta\nu}\hat a(\nu)d\nu,
\ee
where $\hat a$ is the ($\hbar$ independent) Fourier transform of $a$.

Developing (\ref{symm}) we get that $H\psi_{qp}^a=\sum_{k=0}^{k=K}\hbar^{\frac k 2}D_kh(q,p)\psi_{qp}^{b_k}+O(\hbar^{\frac{k+1}2})$, where $b_k\in\mathcal S$ and $D_k$ is an
homogeneous differential operator of order $k$. It is easy to conclude,  thanks to the hypothesis $\frac\hbar{q^2+p^2}=O(1)$,  that
\[
\hbar^{-\frac k 2}(q^2+p^2)^{\frac{M-k}2}=O((q^2+p^2)^{\frac M2}),\ \hbar^{\frac{M+1}2}=O((q^2+p^2)^{\frac M2}).\]
\end{proof}
\end{proof} 
This Proposition is crucial for the rest of this Section, as it allows us to
reduce all  computations to the polynomial setting. For example $A$ may have a
symbol bounded at infinity (class $S(1)$), an assumption which we will need for the application below of Egorov's Theorem), but, with respect to the algebraic equations we will
have to solve, one can consider it as a ``word". In order to simplify our proofs, we will omit the distinction between pseudodifferential operators and their ``word"
approximations.

 \begin{lemma}\label{truc}
Let $A$ be a  WLG($1$)  on $L^2 (\times\bbR ^n\times S^1)$. Let us suppose that $A$ is a symmetric operator.
For $P\in\bbN$ (large), let   
\be\label{ellip}
 A_P:=A+(|D_\theta|^2+|x|^2+|D_x|^2)^{P}
\ee
Then $A_P$ is  an elliptic selfadjoint pseudo-differential operator. Therefore $e^{is\frac {A_P}\hbar}$ is a family of unitary Fourier integral
operators.
\end{lemma}
\begin{proof}
it is enough to observe that $A_P$ is, defined on the domain of $|D_\theta|^2+|x|^2+|D_x|^2$,  a selfadjoint pseudodifferential operator with
symbol of type $S_{1,0}$.
\end{proof}
 
\begin{lemma}\label{quadra}
Let $H_0$ the operator 
\[
H_0=\sum_1^n \theta_i a_ia_i^++D_t
\]
then, if $W$ is a WLG($r$), so is $\frac{[H_0,W]}{i\hbar}$.
\end{lemma}
\begin{proof}
$\frac{[H_0,W]}{i\hbar}=\frac d{ds}e^{isH_0/\hbar}We^{isH_0/\hbar}|_{s=0}$ which, since $H_0$ is quadratic, is the same polynomial as $W$ modulo the substitution 
$a_i\to e^{is}a_i$,
$a_i^+\to e^{-is}a_i$ and shifting of the coefficients  in $t$ by $s$. Therefore the result is immediate.
\end{proof}
More generally:
\begin{lemma}\label{general}
For any   $H$ and $W$ of type WLG($m$) and WLG($r$) respectively, $\frac{[H,W]}{i\hbar}$ is a WLG($m+r-2$).
\end{lemma}
The proof is immediate noting that $[a_i,a^+_j]=\hbar\delta_{ij}$ and that, for any $C^\infty$ function $a(t)$, $[D_t,a]=i\hbar a'$.
 
\vskip 0.5cm
We can now state the main result of this section:
\begin{theorem}\label{direct}
Let $H$ be a (Weyl) pseudo-differential operator on  $L^2 (\bbR ^n\times S^1)$ whose principal symbol if of the form:
\[
H_0(x,\xi;t,\tau)=\sum_1^n\theta_i(x_i^2+\xi_i^2)+\tau + H_2,
\]
where $H_2$ vanishes to third order at $x=\xi=\tau=0$ and $\theta_1,...,\theta_{n}, 2\pi$
are linearly independent over the rationals. Let us define, as before, $P_i=-\hbar^2\frac{\partial^2}{\partial x_i^2}+x_i^2$ and $D_t=-i\hbar\frac{\partial}{\partial t}$.

Then there exists a family of unitary operators $(U_L)_{L=3\dots}$ and constants $(C_L)_{L=3\dots}$, and a $C^\infty$ function $h(p_i,\dots,p_n,\tau,\hbar)$
such that:
\[
||\left(U_L HU_L^{-1}-h(P_1,\dots,P_n,D_t,\hbar)\right)H_\mu||_{L^2 (\bbR ^n\times S^1)}\leq C_L|\mu\hbar|^{\frac{L+1}2}.
\]
\end{theorem}
\vskip 1cm
\begin{proof}
The proof of Theorem \ref{direct} will be a consequence of the following:

\begin{theorem}\label{fond}
Let $H$ be as before, and let $G$ be a WLG($3$). Then there exists a function $G_1(p_1,\dots,p_n,\tau,\hbar)$, a word $F$ and an operator $R$ such that:

i. $\frac {[H,F]}{i\hbar}=G+G_1+R$

ii. $R$ satisfies: $||RH_\mu||=O(|\mu\hbar|^{\frac{L+1}2})$, $\forall\mu\in\bbN^{n+1}$

iii if $G$ is a WLG($\kappa$) so is $F$

iv. if $G$ is a symmetric operator, so is $F$ and $G_1$ is real.
\end{theorem}

Let us first prove that Theorem \ref{fond} implies Theorem \ref{direct}:

by induction, as in the ``classical" case and thanks to Proposition \ref{reduc},  one can assume that $H$ is of the form $H=H_0+G$, where $G$ is a WLG($\kappa$). 
Let us consider the operators $e^{i\frac {F_P}\hbar}He^{-i\frac {F_P}\hbar}$ and $H(s):=e^{is\frac {F_P}\hbar}He^{-is\frac {F_P}\hbar}$, where $F$ satisfies Theorem \ref{fond} and $F_P$
is defined by (\ref{ellip}) for $P$ large enough. By Egorov's Theorem $H(s)$ is a family of pseudodifferential operators. Since we are in an  iterative 
perturbative setting, it is easy to check by taking 
$P$ large enough that we can omit the subscript $P$ in $H(s)$ and let $e^{\pm i\frac {F_P}\hbar}$ stand for $e^{\pm i\frac F\hbar}$ in the rest of the
computation. We have:
\begin{eqnarray}\label{lagrange}
e^{i\frac F\hbar}He^{-i\frac F\hbar}&=& H+\frac{[F,H]}{i\hbar}+\frac{[F,\frac{[F,H]}{i\hbar}]}{i\hbar}+
\frac{[F,[F,[F,\int\limits{_0^1}\int\limits{_0^t}\int\limits{_0^s}H(u)dudsdt]/i\hbar]/i\hbar]}{i\hbar}\nonumber\\
&=& H_0+G+\frac{[F,H]}{i\hbar}+\frac{[F,\frac{[F,H]}{i\hbar}]}{i\hbar}+\tilde R\\
&=&H_0-G_1+R+\frac{[F,\frac{[F,H]}{i\hbar}]}{i\hbar}+\tilde R.\nonumber\end{eqnarray}
Since $H(s)$ is  a pseudodifferential operator,  so is $\int\limits{_0^1}\int\limits{_0^t}\int\limits{_0^s}H(u)dudsdt$. By Proposition \ref{reduc}, Lemma
\ref{general} and Lemma \ref{esti} we have, since $G$ is a WLG($\kappa$),
\[
||\tilde RH_\mu||=O(|\mu\hbar|^{\kappa+1}).
\]
By the same argument, $\frac{[F,\frac{[F,H]}{i\hbar}]}{i\hbar}$ satisfies the same estimate. Developing $\tilde R$ by the Lagrange formula (\ref{lagrange}) to arbitrary order, 
we get, thanks to Lemma \ref{general}, $\tilde R=\tilde G+R$ where $\tilde G$ is a WLG($\kappa+1)$ and 
\[|| RH_\mu||=O(|\mu\hbar|^{\frac{L+1}2}).\]
Therefore, letting $G'=\frac{[F,\frac{[F,H]}{i\hbar}]}{i\hbar}+\tilde G$, we have:
\[
e^{i\frac {F_P}\hbar}He^{-i\frac {F_P}\hbar}=H_0+G_1+G'+R,\] with $G'$  a WLG($\kappa+1$). By induction Theorem \ref{direct} follows.

\underline{Proof of Theorem \ref{fond}}:

let us  first prove the following
\begin{lemma}\label{h0}
Let $H_0$ be as before and let $G$ be a WLG($r$). Then there exists a   WLG($r$) $F$ and  $G_1=G_1(p_1,\dots,p_n,D_t,\hbar)$ 
, such that:

\be\label{cxz}
\frac{[H_0,F]}{i\hbar}=G+G_1
.
\ee
\end{lemma}
\begin{proof}
By Lemma \ref{quadra}, if $F$ is a word, it must be a WLG($r$), since the left hand side of (\ref{cxz}) is WLG($r$).
Let us take the matrix elements of (\ref{cxz}) relating the $H_\mu$s. We get:
\[
-i\Theta.(\mu -\nu)<\mu|F|\nu>= <\mu|G+G_1|\nu>+<\mu|R|\nu>,
\]
where $\Theta.(\mu -\nu):=\sum\limits{_1^n}\theta_i\mu_i+\mu_{n+1}$ and $<\mu|.|\nu>=\left(H_\mu,.H_\nu\right)$. We get immediately that $G_1(\mu\hbar,\hbar)=-<\mu|G|\mu>$. Moreover, 
let us define $F$ by:
\[
<\mu| F|\nu>:=\frac{<\mu|G+G_1|\nu>}{-i\Theta.(\mu -\nu)},
\]
which exists by the non-resonance condition. To show that $F$ is a word one just has to decompose $G=\sum G_l$ in monomial words $G^l=\alpha(t) D_t^jb_1\dots b_m,\
b_i\in\{a_1,a^+_1,\dots,a_n,a^+_n\}$. Then, for each $\nu$ there is only one $\mu$  for which 
$<\mu|G+G_1|\nu>\neq 0$ and the difference $\mu-\nu$ depends obviously only on $G^l$, not on $\nu$. Let us call this difference $\rho_{G_l}$. Then $F$ is given by the sum:
\[
F=\sum\frac 1{-i\Theta.\rho_{G_l}}G_l.
\]

\end{proof}
It is easy to check that one can pass from Lemma \ref{h0} to Theorem \ref{fond} by induction, writing $[H,F+F']=[H,F]+[H_0,F']+[H-H_0,F]+[H-H_0,F']$.
\end{proof}
We will show finally that Theorems \ref{direct} and \ref{dol} are equivalent. Once again we can start by considering an Hamiltonian on $L^2(\bbR^n\times S^1)$ since any Fourier integral
operator $B_\varphi$, as defined in the beginning of the proof of Theorem \ref{dol}, intertwines the original Hamiltonian $H:C^\infty_0(X)\to C^\infty(X)$ of Section \ref{sem}
with a pseudodifferential operator on $L^2(\bbR^n\times S^1)$ satisfying  the hypothesis of Theorem \ref{direct}.

Let us remark first of all that if $U_L=e^{i\frac {W_3}\hbar}e^{i\frac {W_4}\hbar}\dots e^{i\frac {W_L}\hbar}$, all 
$e^{i\frac {W_l}\hbar}$ being Fourier integral operators, then so is $U_L$. Secondly we have  
\begin{proposition}\label{herm}
Let $A$ be a pseudodifferential operator of total Weyl symbol $a(x,\xi,t,\tau,\hbar)$. Then
\[
a\ \mbox{vanishes to infinite order at}\ p=\tau=0 \mbox{ if and only if } ||AH_\mu||_{L^2(\bbR^n\times S^1)}=O(|\mu\hbar|^\infty).
\]
\end{proposition}
\begin{proof}
 the ``if" part  is exactly Proposition \ref{reduc}.
For the ``only if" part let us observe that, if the total symbol didn't vanish to infinite order, then it would contain  terms of the form $\alpha_{kmnr}(t)\hbar^k(x+i\xi)^m(x-i\xi)^n\tau^r$. 
Let us  prove this can't happen in dimension
$1$, the extension to dimension $n$ being straightforward. 

Each term of the form 
$(x+iD_x)^m(x-iD_x)^n=a^m(a^+)^n$ gives rise to an operator $A_{m,n}$ such that: 
\[
A_{m,n}H_\mu=\hbar^{\frac{|m+n|}2}\sqrt{(\mu+1)\dots(\mu+n)(\mu+n-1)\dots(\mu+n-m)}H_{\mu+m-n}\sim |\mu\hbar|^{\frac{m+n}2}H_{\mu+m-n}.
\] 
Therefore $\sum c_{mn}A_{m,n}H_\mu=\sum c_{m,m-l}H_{\mu+l}\sim \sum c_{m,m-l}|\mu\hbar|^{\frac{2m-l}2}H_{\mu+l}$. In particular:
\[
||\sum c_{mn}A_{m,n}H_\mu||^2\sim\sum |\mu\hbar|^{2m-l}\] so             $||\sum c_{mn}A_{mn}H_\mu||=O(|\mu\hbar|^\infty)$ implies $C_{mn}=0$. It is easy to check that 
the same argument is also valid
for any ordered product of $a$'s and $a^+$'s.
\end{proof}

In the next section we will show how the functions $H'$ of  Theorem \ref{direct} and $h$ of Theorem \ref{dol} are related.

\section{Link between the two quantum constructions}
Consider a symbol (on $\bbR^{2n}$) of the form 
\[
h(p_1,...,p_n)
\]
 with $p_i=\frac{\xi_i^2+x_i^2}2$.
The are several ways of quantizing $h$: one of them consists in associating to $h$,
 by the spectral theorem, the operator
 \[
 h(P_1,...,P_n)=h(P)
 \]
 where  $P_i=\frac{-\hbar^2\partial_{x_i}^2+x_i^2}2$.
 Another one is the Weyl quantization procedure.
 
 In this section we want to compute the Weyl symbol $h^{we}$ of $h(P_1,...,P_n)$ and apply the result to the situation of the preceding sections. 
 By the metaplectic invariance of the Weyl quantization and the fact that $h(P_1,...,P_n)$ 
 commutes with all the $P_i$'s we know that $h^{we}$ has the form
 \[
 h^{we}(p_1,...,p_n)=h^{we}(p),
 \]
 that is, is  function of the classical harmonic oscillators $p_i:=\xi_i^2+x_i^2$.
 
 To see how this $h^{we}$ is related to the $h$ above we note that
 $H$ is diagonal on the Hermite basis ${h_j}$. Therefore 
 \[
 h((j+\frac 1 2 )\hbar))=<h_j,Hh_j>=\int h^{we}\left((\frac{x+y}2)^2+\xi^2\right)e^{i\frac{x\xi}\hbar}h_j(x)h_j(\xi)\frac{dxd\xi}{\hbar^{n/2}}.
 \] 
 We now claim
 \begin{proposition}\label{vot}
 let $h$ be either in the Schwartz class, or a polynomial function. Let $\hat h(s)=\frac 1 {(2\pi)^n}\int h(p)e^{-is.p}dp$ be the Fourier transform of $h$.
 Then
 \begin{equation}\label{vof}
 h^{we}(p)=\int\hat h(s)e^{\frac{2i\tan(s\hbar/2).p}\hbar}\Phi(s)ds
 \end{equation}
 where $\tan(s\hbar/2).p$ stands for $\sum_i \tan(s_i\hbar/2) p_i$ and $\Phi(s)=\prod\limits{_{i=1}^n}(1-2i\tan(s_i\hbar/2))$, and where
  (\ref{vof}) has to be interpreted in the sense of distribution, that is, for each $\varphi$ in the Schwartz's
 class of $\bbR$,
 \[
 \int h^{we}(p)\varphi(p)dp=\int \int\hat h(s)e^{\frac{2i\tan(s\hbar/2).p}\hbar}\Phi(s)ds\varphi(p)dp=\int\hat h(s)\Phi(s)
 \hat{\varphi}\left(\frac{2i\tan(s\hbar/2)}\hbar\right)ds.
 \]
 Finally, as $\hbar\to 0$, 
 \begin{equation}\label{vog}
 h^{we}\sim h+\sum_{l=1}^\infty c_l\hbar^{2l}
 \end{equation}
 \end{proposition}
 \begin{proof}
 
 Let  $h(P)=\int \hat{h}(s)e^{is.P}ds$, where $e^{is.P}$ is a zeroth order semiclassical pseudo-differential 
 operator whose Weyl symbol will be computed from its Wick symbol (see \ref{wick} below for the definition). Let us first remark 
 that since $e^{is.P}=\Pi_{i=1}^n e^{is_i.P_i}$  it is enough to prove the Theorem in 
 the one-dimensional case.
 
 Let $\varphi_{x\xi}$ be a coherent state at $(x,\xi)$, that is
 \[
 \varphi_{x\xi}(y)=(\pi\hbar)^{-\frac 14}e^{i\frac{\xi y}\hbar}
 e^{-\frac{(y-x)^2}{2\hbar}}
 \]
 Let $z=\frac{\xi+ix}{\sqrt 2}$, $z'=\frac{\xi'+ix'}{\sqrt 2}$ and 
 $z(t)=\frac{\xi(t)+ix(t)}{\sqrt 2}$. A straightforward computation gives
 \begin{equation}\label{voq}
 \left(\varphi_{x\xi},\varphi_{x'\xi'}\right)=e^\frac{2z\overline{z'}
 -\mid z\mid^2-\mid z'\mid^2}{2\hbar}.
 \end{equation}
 Moreover decomposing $\varphi_{x\xi}$ on the Hermite basis leads to
 \begin{equation}\label{vow}
 e^{isP}\varphi_{x\xi}=e^{i\frac s 2 \hbar}\varphi_{x(\hbar s)\xi(\hbar s)}
 \end{equation}
 where $P=\frac{-\hbar^2\partial_{x}^2+x^2}2$ and $z(t)=e^{it}z$.
 
 The Wick symbol of $e^{isP}$ is defined as 
 \be\label{wick}
 \sigma^{wi}(e^{isP})(x,\xi):=\left(\varphi_{x\xi},
 e^{isP}\varphi_{x\xi}\right)
 \ee
  which, by (\ref{voq}) and (\ref{vow}), is equal to
  \[
  e^{-\frac{1-e^{-i\hbar s}}\hbar \left(\frac{x^2+\xi^2}2\right)+i\frac{s}2\hbar}
 \]
 Moreover, using the Weyl quantization formula, it is immediate to see that the
 Weyl and Wick symbols are related by
 \[
 \sigma^{wi}=e^{-\frac{\hbar\Delta}4}\sigma^{we}
 \]
 where $\Delta=-\left(\frac{\partial^2}{\partial x^2}+\frac{\partial^2}{\partial \xi^2}\right)$.
 
 It is a standard fact the the Wick symbol determines the operator: indeed  the 
 function $e^{\frac{-2z\overline{z'}+|z|^2+|z'|^2}{2\hbar}}(\varphi_{x\xi},
 e^{isP}\varphi_{x'\xi'})$ obviously determines $e^{isP}$. Moreover it  is
 easily seen to be analytic in $z$ and $\overline{z'}$. Therefore it is 
 determined by its values on the diagonal $z=\overline{z'}$ i.e., precisely, the Wick symbol of $e^{isP}$. A straightforward calculation shows that, for 
 $\frac s {2\hbar}\neq
 \frac{(2k+1)\pi}2,\  k\in\bbZ$,
 \be\label{distri}
 (1-2i\tan(s\hbar/2))e^{-\frac{\hbar\Delta}4}e^{\frac{2i\tan(s\hbar/2)\frac{x^2+\xi^2}2}\hbar}=e^{-\frac{1-e^{-i\hbar s}}\hbar \left(\frac{x^2+\xi^2}2\right)+i\frac{s}2\hbar}
 \ee
 This shows that, for $\frac s {2\hbar}\neq
 \frac{(2k+1)\pi}2,\  k\in\bbZ$, we have
 $$\sigma^{we}(e^{isP})(p)=(1-2i\tan(s\hbar/2))e^{\frac{2i\tan(s\hbar/2)\frac{x^2+\xi^2}2}\hbar}.$$ 
 Let us now take $\varphi$ in the Schwartz's
 class of $\bbR$, and $B_\varphi$ be the operator of (total) Weyl symbol $\varphi(\frac{x^2+\xi^2}2)$. 
 Let:
 \[
 f(s):=2\pi\int\sigma^{we}(e^{isP})(p)\varphi(p)pdp=\mbox{Trace}[e^{isP}B_\varphi].
 \] 
 \begin{lemma}
 \[f\in C^\infty(\bbR).\]
 \end{lemma}
 \begin{proof}
 by  metaplectic invariance  
 we know that $B_\varphi$ is diagonal on the Hermite basis. Therefore, $\forall k\in\bbN$,  
 \[
(-i)^k\frac{d^k}{ds^k}f(s):=\mbox{Trace}[e^{isP}P^kB_\varphi]=\sum<h_j,B_\varphi h_j>((j+\frac 1 2)\hbar)^ke^{is(j+\frac 1 2)\hbar}.
 \]
 Since $h_j$ is microlocalized on the circle of radius $(j+\frac 1 2)\hbar$ and $\varphi$ is in the Schwartz class, the sum is absolutely convergent for each $k$.
 \end{proof}
 
 Therefore $f(s)=2\pi\int(1-2i\tan(s\hbar/2))e^{\frac{2i\tan(s\hbar/2)\frac{x^2+\xi^2}2}\hbar}\varphi(p)pdp$ and
 (\ref{distri}) is valid in the sense of distribution (in the variable $p$) for all $s\in\bbR$.
 This expression gives (\ref{vof}) immediately for $h$ in the Schwartz class. When $h$ is a polynomial function it is straightforward to check that, 
 since $\hat h$ is a sum of derivatives of the Dirac mass and $e^{isP}$ is a Weyl operator whose symbol is  $C^\infty$ with respect of $s$, the formula also 
holds in this case.
  The asymptotic expansion (\ref{vog}) is obtained by expanding 
 $e^{\frac{2itg(s\hbar/2)\frac{x^2+\xi^2}2}\hbar}$ near $e^{is\frac{x^2+\xi^2}2}$.
 \end{proof}
 Formula (\ref{vof}) shows clearly that $h^{we}$ depends only on the $\frac{2\pi}\hbar$periodization of
 $\hat h(s) e^{i\frac{s\hbar}2}$, therefore
 \begin{corollary}\label{half}
 $h^{we}$ depends only of the values $h\left( (k+\frac 1 2)\hbar\right),\ k\in\bbN$.
 \end{corollary}
 
 We mention one application of formula (\ref{vof}). Let us suppose first that we have computed the quantum normal form at order $K$, that is
 \[
 h_K(p)=\sum_{|k|=k_1+\dots k_n\leq K}c_kp^k
 :=\sum_{|k|=k_1+\dots k_n\leq K}c_kp_1^{k_1}\dots p_n^{k_n}
 \]
 and let us define $h^{we}_K$ as the Weyl symbol of $h_K(P)$.
 \begin{corollary}
 \begin{eqnarray}
 h^{we}_K(p)&=&\sum_{|k|=k_1+\dots k_n\leq K}c_k\frac{\partial^K}{\partial^k_s}\left(\Phi(s)e^{\frac{2i\tan(s\hbar/2)p}\hbar}\right)|_{s=0}\nonumber\\
 &:=&\sum_{|k|=k_1+\dots k_n\leq
 K}c_k\frac{\partial^K}{\partial^{k_1}_{s_1}\dots\partial^{k_n}{s_n}}\left(\Phi(s)e^{2i\frac{\tan(s_1\hbar/2)p_1+\dots+\tan(s_n\hbar/2)p_n}\hbar}\right)|_{s=0}.\nonumber
 \end{eqnarray}
 \end{corollary}
 
 Let us come back now to the comparison between the two constructions of
 Sections 2 and 3.
 
 Clearly the "$\theta$" part doesn't play any role, as the Weyl quantization of
 any function $f(\tau)$ is exactly $f(D_\theta)$. therefore we have the following
 
 \begin{theorem}
 The functions $H'$ of Theorem \ref{dol} and $h$ of Theorem \ref{direct} are
 related by the formula
 \[
 H'\left(P_1,...,P_n,D_t,\hbar\right)=
 \int\hat H^*(s, D_t,\hbar)e^{\frac{2i\tan(s\hbar/2).p}\hbar}\Phi(s)ds
 \]
 where $\hat H^*$ is the Fourier transform of $H^*$ with respect to the
 variables $p_i$. In particular
 \[
 H'-H^*=O(\hbar^2).
 \] 
 \end{theorem}
 \begin{proof}
 the proof follows immediately from Proposition \ref{vot}, and the unicity of the (quantum)
 Birkhoff normal form.
 \end{proof}


\section{The computation of the semiclassical Birkhoff canonical form from the 
asymptotics of the trace formula}
Let $X$ and $H$ be as in the introduction. Let $\gamma$ be a periodic trajectory of the vector
 field (\ref{col}) of period $2\pi$.
 
 For $l\in\bbZ$ let $\varphi_l$ be a Schwartz function on the real line whose Fourier transform 
 $\hat{\varphi_l}$ is supported in a neighborhood of $2\pi l$ containing no 
 other period of (\ref{col}). The semiclassical trace formula gives an asymptotic expansion
 for $\mbox{Trace}\ \varphi_l\left(\frac{H-E}\hbar\right)$ of the form:
 \begin{equation}\label{gol}
 \mbox{Trace}\ \varphi_l\left(\frac{H-E}\hbar\right)
 \sim\sum_{m=0}^\infty d_l^m\hbar^m
 \end{equation}
 where the $d_l$'s are distributions acting on $\hat{\varphi_l}$ with support concentrated at 
  $\{2\pi l\}$.
 
 We will show that the knowledge of the $d_l$s determine the quantum semiclassical Birkhoff form of Section 2, 
 and therefore the classical one.
 
 Let us first rewrite the l.h.s of (\ref{gol}) as 
 \begin{equation}\label{gon}
 \mbox{Trace}\left(\int \hat{\varphi}(t)e^{it\frac{H-E}\hbar}dt\right)
 \end{equation}
 Since $\hat{\varphi}$ is supported near a single period of (\ref{col}) we know
 from the general theory of Fourier integral operators that one can 
 microlocalize (\ref{gol}) near $\gamma$ modulo error term of order $O(\hbar^\infty)$.
 
 Therefore we can conjugate (\ref{gon}) by the semiclassical Fourier integral 
 operator $A_\varphi$ of Theorem \ref{dol}. This leads to the computation of
 \begin{eqnarray}\label{gob}
 &\mbox{Trace}\left(A_\varphi\int \hat{\varphi}(t)e^{it\frac{H-E}\hbar}dt
 A_\varphi^{-1}\right)\\
 =& \mbox{Tr}\left(\int \hat{\varphi}(t)\rho(P_1+...+P_n)
 e^{it\frac{H'(P_1+...+P_n,D_t,\hbar)+H''-E}\hbar}dt\right)\nonumber
 \end{eqnarray}
 where $\rho\in C^\infty_0(\bbR^{n+1})$ with $\rho=1$ in a neighborhood of ${0}$ 
 and $\mbox{Tr}$ stands for the Trace in $L^2(\bbR^n\times S^1)$.
 
 Let us write $H'(P_1+...+P_n,D_t,\hbar)$ as 
 \begin{equation}\label{goc}
 E+D_t +\sum\theta_iP_i+\sum_{r\in \bbN^n,s\in\bbZ}c_{r,s}(\hbar)P^rD_t^s.
 \end{equation}
 We will first prove
 \begin{proposition}\label{gof}
 let $g^l_{r,s}(t,\theta)$ be the function defined by
 \begin{equation}\label{god}
 g^l_{r,s}(t,\theta)=\left(-i\frac\partial{t\partial \theta}\right)^r
 \left(-i\frac\partial{\partial t}\right)^s
 \left[ \frac{e^{it\frac{\theta_1+...+\theta_n}2}}{\Pi_i(1-e^{it\theta_i})}t\hat{\varphi}(t)\right]
 \end{equation}
 Let us fix $l\in\bbZ$. Then the knowledge of all the $d^m_l$s for $m<M$ in (\ref{gol}) 
 determines the following quantities
 \begin{equation}\label{gos}
 \sum_{\mid r\mid+s=m}c_{r,s}(\hbar)g^l_{r,s}(2\pi l,\theta)
 \end{equation}
 for all $m<M$.
 \end{proposition}
 \begin{proof}
 the r.h.s. of (\ref{gob}) can be computed thanks to (\ref{goc}) using
 \begin{eqnarray}
 \mbox{spectrum} \ P_i&=&\{(\mu_i+\frac 12)\hbar,\ \mu_i\in\bbN\}\nonumber\\
 \mbox{spectrum} \ D_t\ &=& \{\nu\hbar,\ n\in\bbZ\}\nonumber
 \end{eqnarray}
 Thus the r.h.s of (\ref{gob}) can be written as
 \begin{equation}
 \int\hat{\varphi_l}(t)\sum_{\mu,\nu}\rho\left(\left(\mid\mu\mid+\frac n2\right)\hbar\right)
 e^{it\left[\nu+\theta.(\mu+\frac 12)\right]}
 \sum_{k=0}^\infty \frac{(it)^k}{k!}
 \left(\sum_{r,s}c_{r,s}(\hbar)\left(\mu+\frac 12\right)^r\nu^s\hbar^{\mid\mu\mid+s-1}\right)^kdt
 \end{equation}
 since the support of $\hat\varphi_l$ contains only one period, and therefore
 the trace can be microlocalized infinitely close to the periodic trajectory, making the role of
 $H''$ inessential.
 
 Using the following remark of S. Zelditch:
 \begin{equation}\nonumber
 (\mu+\frac 12)^r\nu^s=\left(-i\frac\partial{t\partial \theta}\right)^r
 \left(-i\frac\partial{\partial t}\right)^se^{it\left[\nu+\theta.(\mu+\frac 12\right]}
 \end{equation}
 we get, mod($\hbar^\infty$),
 \begin{equation}\label{got}
 \int\hat{\varphi_l}(t)\sum_{k=0}^\infty \frac{(it)^k}{k!}
 \left(\sum_{r,s}\hbar^{\mid\mu\mid+s-1}c_{r,s}(\hbar)
 \left(-i\frac\partial{t\partial \theta}\right)^r\left(-i\frac\partial{\partial t}\right)^s\right)^k
 \sum_{\mu,\nu}e^{it\left[\nu+\theta.(\mu+\frac 12\right]}dt.
 \end{equation}
 Since $\sum_{\nu\in\bbZ}e^{it\nu}=2\pi\sum_l\delta(t-2\pi l)$ 
 , and 
 $\sum_{\mu\in\bbN^n}e^{it\theta.\left(\mu+\frac 12\right)}=
 \frac{e^{it\frac{\theta_1+...+\theta_n}2}}{\Pi_i\left(1-e^{it\theta_i}\right)}$,
 together with the fact that $\hat{\varphi}$ is supported near $2\pi l$, we get that (\ref{got}) is equal to
 \begin{equation}\nonumber
 2\pi\left[\sum_{k=0}^\infty \frac{(i)^k}{k!}
 \left(\sum_{r,s}\hbar^{\mid\mu\mid+s-1}c_{r,s}(\hbar)
 \left(-i\frac\partial{t\partial \theta}\right)^r\left(-i\frac\partial{\partial t}\right)^s\right)^k
 \left(t^k\hat{\varphi_l}(t)
 \frac{e^{it\frac{\theta_1+...+\theta_n}2}}{\Pi_i\left(1-e^{it\theta_i}\right)}\right)
 \right]_{t=2\pi l}.
 \end{equation}
 Rearranging terms in increasing powers of $\hbar$ shows that 
 the quantities (\ref{gos}) can be computed recursively.
 \end{proof}
 The fact that one can compute the $c_{r,s}(\hbar)$ from the quantities (\ref{gos}) 
 is an easy consequence of the rational independence of the $\theta_i$s and the Kronecker theorem, and is exactly the same as in \cite{gu}.

 \

\end{document}